\documentclass{article}

\usepackage{amssymb,amsthm,amsmath}

\newcommand{\ZZ}{\mathbb{Z}}

\newcommand{\FF}{\mathbb{F}}

\newtheorem{theorem}{Theorem}[section]

\newtheorem{proposition}[theorem]{Proposition}
\newtheorem{corollary}[theorem]{Corollary}

\theoremstyle{definition}

\newtheorem{remark}[theorem]{Remark}
\usepackage[square,numbers,sort&compress]{natbib}
\usepackage{hyperref}
\usepackage{breakurl}
\usepackage{url}
\hypersetup{pdfpagemode=none,colorlinks,citecolor=blue,filecolor=black,linkcolor=black,urlcolor=blue}
\newcommand{\vct}[1]{\mathbf{#1}}
\newcommand{\nth}{^{\text{th}}}
\DeclareMathOperator{\moddec}{mod}
\renewcommand{\mod}[1]{\,(\moddec #1)}
\usepackage{enumitem}
\usepackage{nameref}
\makeatletter
\newcommand{\clabel}[2]{\protected@write \@auxout {}{\string \newlabel {#1}{{#2}{\thepage}{#2}{#1}{}} }\hypertarget{#1}{}}
\makeatother
\newcommand{\lb}{\allowbreak}
\usepackage{etoolbox}
\newcommand{\breaklist}[2][,\lb]{\def\nextitem{\def\nextitem{#1}}\renewcommand*{\do}[1]{\nextitem{##1}}\docsvlist{#2}}
\DeclareMathOperator{\cirdec}{circ}		
	
\newcommand{\cir}[1]{\cirdec(\breaklist{#1})}

\DeclareMathOperator{\rankdec}{rank}
\newcommand{\rank}[1]{\rankdec(\breaklist{#1})}
\newcommand{\floor}[1]{\lfloor#1\rfloor}
\usepackage{empheq}

\DeclareMathOperator{\autdec}{Aut}
\newcommand{\aut}[1]{\autdec(#1)}
\usepackage{booktabs}
\usepackage{upgreek}
\newcommand{\vctg}[1]{\boldsymbol{#1}}
\makeatletter
\renewcommand*\env@matrix[1][*\c@MaxMatrixCols c]{\hskip -\arraycolsep\let\@ifnextchar\new@ifnextchar\array{#1}}
\makeatother
\usepackage[font=footnotesize,labelfont=bf,textfont=normalfont,margin=0cm]{caption}
\usepackage{adjustbox}
\newcommand{\pcir}[2]{\cirdec_{#1}(\breaklist{#2})}
\providecommand{\keywords}[1]{\small\textit{Keywords}: #1}
\providecommand{\msc}[1]{\small\textit{2020 MSC}: #1}

\title{Binary self-dual codes of various lengths with new weight enumerators from a modified bordered construction and neighbours}

\author{J. Gildea, A. Korban and A. M. Roberts\\
Department of Mathematical and Physical Sciences\\
University of Chester\\
Exton Park\\
Chester CH1 4AR\\
United Kingdom\\
{}\\
A. Tylyshchak\\
Department of Algebra\\
Uzhgorod National University\\
Uzhgorod\\
Ukraine
}
\date{}

\begin{document}

\maketitle


\msc{94B05, 15B10, 15B33}

\let\thefootnote\relax\footnote{E-mail addresses: \href{mailto:j.gildea@chester.ac.uk}{j.gildea@chester.ac.uk} (J. Gildea), \href{mailto:adrian3@windowslive.com}{adrian3@windowslive.com} (A. Korban), \href{mailto:adammichaelroberts@outlook.com}{adammichaelroberts@outlook.com} (A. M. Roberts),
\href{mailto:}{alxtlk@bigmir.net} (A. Tylyshchak)
}

\begin{abstract}
In this work, we define a modification of a bordered construction for self-dual codes which utilises $\lambda$-circulant matrices. We provide the necessary conditions for the construction to produce self-dual codes over finite commutative Frobenius rings of characteristic 2. Using the modified construction together with the neighbour construction, we construct many binary self-dual codes of lengths 54, 68, 82 and 94 with weight enumerators that have previously not been known to exist.
\end{abstract}

\section{Introduction}

The class of self-dual codes is widely researched in coding theory. Not only are they rich in mathematical theory, but they also have close relationships to other mathematical structures such as lattices, designs and modular forms. Much effort has been invested into developing techniques for constructing self-dual codes and particularly extremal binary self-dual codes, i.e. binary self-dual codes whose minimum distance meets a specific bound. 

Two of the most famous techniques include the double circulant and bordered double circulant constructions. All extremal double circulant and bordered double circulant binary self-dual codes have been classified up to length 96 \cite{R-056,R-057,R-058,R-059}. Another well-known technique is the four circulant construction, which was introduced in \cite{R-014} and has since been used to great effect in producing binary self-dual codes \cite{R-136,R-184,R-123,R-176,R-097,AMR1}. A substantial amount of work has been done on constructing self-dual codes having an automorphism of odd prime order \cite{R-193,R-072,R-071,R-085,R-046,R-125}. Recently, a strong connection between group rings and self-dual codes was established \cite{R-008} which has been utilised to develop a number of different techniques for constructing extremal binary self-dual codes \cite{R-043,R-113,R-096,R-067}.

Bordered matrix constructions have also proven to be effective techniques for constructing binary self-dual codes \cite{R-043,R-098,R-067,R-111,R-108,AMR4}. In this work, we present a new bordered matrix construction derived as a modification of a construction recently given in \cite{AMR4}. By applying this new construction, we obtain many extremal, optimal and best known binary self-dual codes that have previously not been known to exist. In particular, together with the neighbour method, we construct binary self-dual codes of lengths 54, 68, 82 and 94 with weight enumerator parameters of previously unknown values. We also provide the conditions needed by the construction to produce self-dual codes over a finite commutative Frobenius ring of characteristic 2.

The paper is organised as follows. In Section \ref{section-2}, we give preliminary definitions and results on self-dual codes, the alphabets we use and special matrices which we use in this work. In Section \ref{section-3}, we present the new construction and prove under what conditions it produces self-dual codes over finite commutative Frobenius rings of characteristic 2. In Section \ref{section-4}, we apply the new construction and the neighbour construction to obtain the new self-dual codes which we also tabulate. We finish with some concluding remarks and discussion of possible directions for future work.

\section{Preliminaries}\label{section-2}

\subsection{Self-Dual Codes}

Let $R$ be a commutative Frobenius ring (see \cite{B-014} for a full description of Frobenius rings and codes over Frobenius rings). Throughout this work, we always assume $R$ has unity. A code $\mathcal{C}$ of length $n$ over $R$ is a subset of $R^n$ whose elements are called codewords. If $\mathcal{C}$ is a submodule of $R^n$, then we say that $\mathcal{C}$ is linear. Let $\mathbf{x},\mathbf{y}\in R^n$ where $\mathbf{x}=(x_1,x_2,\dots,x_n)$ and $\mathbf{y}=(y_1,y_2,\dots,y_n)$. The (Euclidean) dual $\mathcal{C}^{\bot}$ of $\mathcal{C}$ is given by
	\begin{equation*}
	\mathcal{C}^{\bot}=\{\mathbf{x}\in R^n: \langle\mathbf{x},\mathbf{y}\rangle=0,\forall\mathbf{y}\in\mathcal{C}\},
	\end{equation*}	
where $\langle\cdot,\cdot\rangle$ denotes the Euclidean inner product defined by
	\begin{equation*}
	\langle\mathbf{x},\mathbf{y}\rangle=\sum_{i=1}^nx_iy_i.
	\end{equation*}

We say that $\mathcal{C}$ is self-orthogonal if $\mathcal{C}\subseteq \mathcal{C}^\perp$ and self-dual if $\mathcal{C}=\mathcal{C}^{\bot}$.

An upper bound on the minimum (Hamming) distance of a doubly-even (Type II) binary self-dual code was given in \cite{R-116} and likewise for a singly-even (Type I) binary self-dual code in \cite{R-115}. Let $d_{\text{I}}(n)$ and $d_{\text{II}}(n)$ be the minimum distance of a Type I and Type II binary self-dual code of length $n$, respectively. Then
	\begin{equation*}
	d_{\text{II}}(n)\leq 4\floor{n/24}+4
	\end{equation*}
and
	\begin{equation*}
	d_{\text{I}}(n)\leq
	\begin{cases}
	4\floor{n/24}+2,& \text{if }n\equiv 0\pmod{24},\\
	4\floor{n/24}+4,& \text{if }n\not\equiv 22\pmod{24},\\
	4\floor{n/24}+6,& \text{if }n\equiv 22\pmod{24}.
	\end{cases}
	\end{equation*}

A self-dual code whose minimum distance meets its corresponding bound is called \textit{extremal}. A self-dual code with the highest possible minimum distance for its length is said to be \textit{optimal}. Extremal codes are necessarily optimal but optimal codes are not necessarily extremal. A \textit{best known} self-dual code is a self-dual code with the highest known minimum distance for its length.

\subsection{Alphabets}

In this paper, we consider the alphabets $\FF_2$ and $\FF_2+u\FF_2$.

Define
	\begin{equation*}
	\FF_2+u\FF_2=\{a+bu:a,b\in\FF_2,u^2=0\}.
	\end{equation*}

Then $\FF_2+u\FF_2$ is a commutative ring of order 4 and characteristic 2 such that $\FF_2+u\FF_2\cong\FF_2[u]/\langle u^2\rangle$. 

We recall the following Gray map from \cite{R-117}
	\begin{align*}
	\varphi_{\FF_2+u\FF_2}&:(\FF_2+u\FF_2)^n\to\FF_2^{2n}\\
        &\quad a+bu\mapsto(b,a+b),\,a,b\in\FF_2^{n}.
	\end{align*}

Note that this Gray map preserves orthogonality. The Lee weight of a codeword is defined to be the Hamming weight of its binary image under the aforementioned Gray map. A self-dual code in $R^n$ where $R$ is equipped with a Gray map to the binary Hamming space is said to be of Type II if the Lee weights of all codewords are multiples of 4, otherwise it is said to be of Type I.
	\begin{proposition}\textup{(\cite{R-117})}\label{proposition-1}
		Let $\mathcal{C}$ be a code over $\FF_2+u\FF_2$. If $\mathcal{C}$ is self-orthogonal, then $\varphi_{\FF_2+u\FF_2}(\mathcal{C})$ is self-orthogonal. The code $\mathcal{C}$ is a Type I (resp. Type II) code over $\FF_2+u\FF_2$ if and only if $\varphi_{\FF_2+u\FF_2}(\mathcal{C})$ is a Type I (resp. Type II) code over $\FF_2$. The minimum Lee weight of $\mathcal{C}$ is equal to the minimum Hamming weight of $\varphi_{\FF_2+u\FF_2}(\mathcal{C})$.
	\end{proposition}

The next corollary follows directly from Proposition \ref{proposition-1}.
	\begin{corollary}
		Let $\mathcal{C}$ be a self-dual code over $\FF_2+u\FF_2$ of length $n$ and minimum Lee distance $d$. Then $\varphi_{\FF_2+u\FF_2}(\mathcal{C})$ is a binary self-dual $[2n,n,d]$ code. Moreover, the Lee weight enumerator of $\mathcal{C}$ is equal to the Hamming weight enumerator of $\varphi_{\FF_2+u\FF_2}(\mathcal{C})$. If $\mathcal{C}$ is a Type I (resp. Type II) code, then $\varphi_{\FF_2+u\FF_2}(\mathcal{C})$ is a Type I (resp. Type II) code.
	\end{corollary}
	
\subsection{Special Matrices}

We now define and discuss the properties of some special matrices which we use in our work. Let $\vct{a}=(a_0,a_1,\ldots,a_{n-1})\in R^n$ where $R$ is a commutative ring and let
	\begin{equation*}
	A=\begin{pmatrix}
	a_0 & a_1 & a_2 & \cdots & a_{n-1}\\
	\lambda a_{n-1} & a_0 & a_1 & \cdots & a_{n-2}\\
	\lambda a_{n-2} & \lambda a_{n-1} & a_0 & \cdots & a_{n-3}\\
	\vdots & \vdots & \vdots & \ddots & \vdots\\
	\lambda a_1 & \lambda a_2 & \lambda a_3 & \cdots & a_0
	\end{pmatrix},
	\end{equation*}
where $\lambda\in R$. Then $A$ is called the $\lambda$-circulant matrix generated $\vct{a}$, denoted by $A=\pcir{\lambda}{\vct{a}}$. If $\lambda=1$, then $A$ is called the circulant matrix generated by $\vct{a}$ and is more simply denoted by $A=\cir{\vct{a}}$. If we define the matrix
	\begin{equation*}
	P_{\lambda}=\begin{pmatrix}
	\vct{0} & I_{n-1}\\
	\lambda & \vct{0}
	\end{pmatrix},
	\end{equation*}
then it follows that $A=\sum_{i=0}^{n-1}a_iP_{\lambda}^i$. Clearly, the sum of any two $\lambda$-circulant matrices is also a $\lambda$-circulant matrix. If $B=\pcir{\lambda}{\vct{b}}$ where $\vct{b}=(b_0,b_1,\ldots,b_{n-1})\in R^n$, then $AB=\sum_{i=0}^{n-1}\sum_{j=0}^{n-1}a_ib_jP_{\lambda}^{i+j}$. Since $P_{\lambda}^n=\lambda I_n$ there exist $c_k\in R$ such that $AB=\sum_{k=0}^{n-1}c_kP_{\lambda}^k$ so that $AB$ is also $\lambda$-circulant. In fact, it is true that
	\begin{equation*}
	c_{k}=\sum_{\substack{[i+j]_n=k\\i+j<n}}a_ib_j+\sum_{\substack{[i+j]_n=k\\i+j\geq n}}\lambda a_ib_j=\vct{x}_1\vct{y}_{k+1}
	\end{equation*}
for $k\in\{0,\ldots,n-1\}$, where $\vct{x}_i$ and $\vct{y}_i$ respectively denote the $i\nth$ row and column of $A$ and $B$ and $[i+j]_n$ denotes the smallest non-negative integer such that $[i+j]_n\equiv i+j\mod{n}$. From this, we can see that $\lambda$-circulant matrices commute multiplicatively and in fact the set of $\lambda$-circulant matrices over a commutative ring of fixed size is itself a commutative ring. Moreover, if $\lambda$ is a unit in $R$, then $A^T$ is $\lambda^{-1}$-circulant such that $A^T=a_0I_n+\lambda\sum_{i=1}^{n-1}a_{n-i}P_{\lambda^{-1}}^i$. It follows then that $AA^T$ is $\lambda$-circulant if and only if $\lambda$ is involutory in $R$, i.e. $\lambda^2=1$.

\section{The Construction}\label{section-3}
In this section, we present our technique for constructing self-dual codes. We will hereafter always assume $R$ is a finite commutative Frobenius ring of characteristic 2.
	\begin{theorem}\label{theorem-1}
	Let $n\in\ZZ^+$ and let
		\begin{align*}
			G=
			\begin{pmatrix}[c|c|cc]
				\vct{v} & \vct{0} & \xi_5 & \xi_6\\\midrule
				I_{2n} & X & \vct{v}^T & \vct{v}^T
			\end{pmatrix},\quad\text{where }
			X=\begin{pmatrix}
			AC & B\\
			B^TC & A^T
			\end{pmatrix}
		\end{align*}
	where $\vct{v}=(\vct{v}_1,\vct{v}_2)$ such that
		\begin{align*}
			\vct{v}_1&=(\vctg{\xi}_1,\xi_2)\in R^n,\\
			\vct{v}_2&=(\vctg{\xi}_3,\xi_4,\xi_4)\in R^n,
		\end{align*}
	with $\vctg{\xi}_{2j-1}=(\xi_{2j-1},\xi_{2j-1},\ldots,\xi_{2j-1})\in R^{n-j}$ for $j\in\{1,2\}$ and $\xi_i\in R$ for $i\in\{1,\ldots,6\}$ also with $A=\pcir{\lambda}{\vct{a}}$, $B=\pcir{\lambda}{\vct{b}}$ and $C=\pcir{\mu}{\vct{c}}$ for $\vct{a},\vct{b},\vct{c}\in R^n$ and $\lambda,\mu\in R:\lambda^2=\mu^2=1$. Then $G$ is a generator matrix of a self-dual code of length $2(2n+1)$ if and only if
		\begin{empheq}[left=\empheqlbrace]{align*}
			AA^T+BB^T&=I_n,\\
			CC^T&=I_n,\\
			\xi_{n'}^2+\xi_2^2+\xi_5^2+\xi_6^2&=0,\\
			\xi_j(\xi_5+\xi_6+1)&=0,\quad j\in\{1,\ldots,4\},
		\end{empheq}
	and the free rank of $(\vct{v}_1A+\vct{v}_2B^T,\vct{v}_1B+\vct{v}_2A^T,\xi_5,\xi_6)$ is 1, where $n'=2[n]_2+1$.
	\end{theorem}
		\begin{proof}
			First, let us determine the conditions required for $G$ to be a generator matrix of a self-orthogonal code. We have
				\begin{equation*}
					GG^T=
					\begin{pmatrix}[c|c|cc]
						\vct{v} & \vct{0} & \xi_5 & \xi_6\\\midrule
						I_{2n} & X & \vct{v}^T & \vct{v}^T
					\end{pmatrix}
					\begin{pmatrix}[c|c]
						\vct{v}^T & I_{2n}\\\midrule
						\vct{0} & X^T\\\midrule
						\xi_5 & \vct{v}\\
						\xi_6 & \vct{v}
					\end{pmatrix}
					=
					\begin{pmatrix}
						g_{1,1} & g_{1,2}\\
						g_{1,2}^T & g_{2,2}
					\end{pmatrix}
				\end{equation*}
			where
				\begin{align*}
					g_{1,1}&=\vct{v}\vct{v}^T+\xi_5^2+\xi_6^2,\\
					g_{1,2}&=(\xi_5+\xi_6+1)\vct{v},\\
					g_{2,2}&=XX^T+I_{2n}+2\vct{v}\vct{v}^T
				\end{align*}
			so that $GG^T=\vct{0}$ if and only if $g_{1,1}=0$, $g_{1,2}=\vct{0}$ and $g_{2,2}=\vct{0}$. Since $R$ is of characteristic 2, we have
				\begin{align*}
					g_{1,1}&=\vct{v}\vct{v}^T+\xi_5^2+\xi_6^2\\
				    &=(n-1)\xi_1^2+\xi_2^2+(n-2)\xi_3^2+2\xi_4^2+\xi_5^2+\xi_6^2\\
				    &=(n-1)\xi_1^2+\xi_2^2+(n-2)\xi_3^2+\xi_5^2+\xi_6^2\\
				    &=\begin{cases}
					    \xi_1^2+\xi_2^2+\xi_5^2+\xi_6^2,&n\text{ is even},\\
				    	\xi_3^2+\xi_2^2+\xi_5^2+\xi_6^2,&n\text{ is odd},
				    \end{cases}\\
				    &=\xi_{n'}^2+\xi_2^2+\xi_5^2+\xi_6^2,
				\end{align*}		
			where $n'=2[n]_2+1$ (recall that $[n]_2$ is the smallest non-negative integer such that $[n]_2\equiv n\mod{2}$), 
			so $g_{1,1}=0$ if and only if $\xi_{n'}^2+\xi_2^2+\xi_5^2+\xi_6^2=0$. We also have $2\vct{v}\vct{v}^T=\vct{0}$, so $g_{2,2}=\vct{0}$ if and only if $XX^T=I_{2n}$. By Lemma 3.1 of \cite{AMR4}, $XX^T=I_{2n}$ if and only if $AA^T+BB^T=I_n$ and $CC^T=I_n$. Finally, we see that $g_{1,2}=\vct{0}$ if and only if $\xi_j(\xi_5+\xi_6+1)=0$ for $j\in\{1,\ldots,4\}$.
			
			Assume now that $G$ is a matrix of a self-orthogonal code. We need to prove that the free rank of $G$ is $2n+1$ if and only if the free rank of $(\vct{v}_1A+\vct{v}_2B^T,\vct{v}_1B+\vct{v}_2A^T,\xi_5,\xi_6)$ is 1. The free rank of $G$ is unchanged by elementary row (or column) operations and premultiplication (or postmultiplication) by an invertible matrix of appropriate size. Let $\tilde{G}=GM$ where
				\begin{align*}
					M=\begin{pmatrix}[c|c|c]
						I_{2n} & X & \begin{matrix}\vct{v}^T & \vct{v}^T\end{matrix}\\\midrule
						\vct{0} & I_{2n} & \vct{0}\\\midrule
						\vct{0} & \vct{0} & I_2
					\end{pmatrix}.
				\end{align*}
				
			Let $\rank{}$ denote the free rank of a matrix over $R$. It is clear that $M$ is invertible and hence $\rank{\tilde{G}}=\rank{G}$. We have that
				\begin{align*}
					GM&=
					\begin{pmatrix}[c|c|cc]
						\vct{v} & \vct{0} & \xi_5 & \xi_6\\\midrule
						I_{2n} & X & \vct{v}^T & \vct{v}^T
					\end{pmatrix}
					\begin{pmatrix}[c|c|cc]
						I_{2n} & X & \begin{matrix} \vct{v}^T & \vct{v}^T\end{matrix}\\\midrule
						\vct{0} & I_{2n} & \vct{0}\\\midrule
						\vct{0} & \vct{0} & I_2
					\end{pmatrix}\\&=
					\begin{pmatrix}
						\vct{v} & \vct{v}X & \vct{v}(\vct{v}^T,\vct{v}^T)+(\xi_5,\xi_6)\\
						I_{2n} & \vct{0} & \vct{0}
					\end{pmatrix}\\&=
					\begin{pmatrix}
						\vct{v} & \vct{v}X & (\vct{v}\vct{v}^T+\xi_5,\vct{v}\vct{v}^T+\xi_6)\\
						I_{2n} & \vct{0} & \vct{0}
					\end{pmatrix}.
				\end{align*}
				
			Let $r=\rank{(\vct{v}X,\vct{v}\vct{v}^T+\xi_5,\vct{v}\vct{v}^T+\xi_6)}$. Then $\rank{\tilde{G}}=2n+1$ if and only if $r=1$. We see that
				\begin{align*}
					\vct{v}X&=
					(\vct{v}_1,\vct{v}_2)
					\begin{pmatrix}
						AC & B\\
						B^TC & A^T
					\end{pmatrix}\\&=
					((\vct{v}_1A+\vct{v}_2B^T)C,\vct{v}_1B+\vct{v}_2A^T).
				\end{align*}
			and
				\begin{align*}
					(\vct{v}\vct{v}^T+\xi_5,\vct{v}\vct{v}^T+\xi_6)=(\xi_5+\xi_{n'}^2+\xi_2^2,\xi_6+\xi_{n'}^2+\xi_2^2).
				\end{align*}
				
			Since $G$ is a generator matrix of a self-orthogonal code, we have $\xi_{n'}^2+\xi_2^2+\xi_5^2+\xi_6^2=0$ so that $\xi_{n'}^2+\xi_2^2=\xi_5^2+\xi_6^2$. By elementary column operations we obtain
				\begin{align*}
					r&=\rank{(\vct{v}X,\vct{v}\vct{v}^T+\xi_5,\vct{v}\vct{v}^T+\xi_6)}\\&=
					\rank{(\vct{v}X,\xi_5+\xi_{n'}^2+\xi_2^2,\xi_6+\xi_{n'}^2+\xi_2^2)}\\&=
					\rank{(\vct{v}X,\xi_5+\xi_5^2+\xi_6^2,\xi_6+\xi_5^2+\xi_6^2)}\\&=
					\rank{(\vct{v}X,\xi_5+\xi_5^2+\xi_6^2,\xi_5+\xi_6)}\\&=
					\rank{(\vct{v}X,\xi_5+\xi_5^2+\xi_6^2+(\xi_5+\xi_6)^2,\xi_5+\xi_6)}\\&=
					\rank{(\vct{v}X,\xi_5,\xi_5+\xi_6)}\\&=
					\rank{(\vct{v}X,\xi_5,\xi_6)}.
				\end{align*}
				
			We also have that $CC^T=I_n$ so that $C$ is invertible. Thus, we get
				\begin{align*}
					r&=\rank{(\vct{v}X,\xi_5,\xi_6)}\\&=
					\rank{((\vct{v}_1A+\vct{v}_2B^T)C,\vct{v}_1B+\vct{v}_2A^T,\xi_5,\xi_6}\\&=
					\rank{(\vct{v}_1A+\vct{v}_2B^T,\vct{v}_1B+\vct{v}_2A^T,\xi_5,\xi_6)}
				\end{align*}
			and so $\rank{G}=2n+1$ if and only if $\rank{(\vct{v}_1A+\vct{v}_2B^T,\vct{v}_1B+\vct{v}_2A^T,\xi_5,\xi_6)}=1$.
		\end{proof}

\section{Results}\label{section-4}

In this section, we apply Theorem \ref{theorem-1} to obtain many new extremal, optimal and best known binary self-dual codes. In particular, we obtain 7 new extremal codes of length $68$, 18 new best known codes of length $82$ and 12 new best known codes of length $94$.

	\begin{remark}\label{remark-1}
	Two binary self-dual codes of length $2n$ are said to be neighbours if their intersection has dimension $n-1$. Let $\mathcal{C}^*$ be a binary self-dual code of length $2n$ and let $\vct{x}\in\FF_2^n\setminus \mathcal{C}^*$. Then $\mathcal{C}=\langle\langle\vct{x}\rangle^{\perp}\cap\mathcal{C}^*,\vct{x}\rangle$ is a neighbour of $\mathcal{C}^*$, where $\langle\vct{x}\rangle$ denotes the code generated by $\vct{x}$.
	\end{remark}
	
Using Remark \ref{remark-1}, we obtain one new optimal code of length $54$ and 7 new extremal codes of length $68$ as neighbours of codes constructed by applying Theorem \ref{theorem-1}.

We conduct the search for these codes using MATLAB and Magma \cite{Magma} and determine their properties using Q-extension \cite{Q-extension} and Magma. In MATLAB, we employ an algorithm which randomly searches for the construction parameters that satisfy the necessary and sufficient conditions stated in Theorem \ref{theorem-1}. For such parameters, we then build the corresponding binary generator matrices and print them to text files. We then use Q-extension to read these text files and determine the minimum distance and partial weight enumerator of each corresponding code. Furthermore, we determine the automorphism
group order of each code using Magma. Magma is used to search for neighbours as described in Remark \ref{remark-1}. A database of generator matrices of the new codes is given online at \cite{GMD5}. The database is partitioned into text files (interpretable by Q-extension) corresponding to each code type. In these files, specific properties of the codes including the construction parameters, weight enumerator parameter values and automorphism group order are formatted as comments above the generator matrices. Partial weight enumerators of the codes are also formatted as comments below the generator matrices. Table \ref{table-1} gives the quaternary notation system we use to represent elements of $\FF_2+u\FF_2$.

	\begin{table}
	\caption{Quaternary notation system for elements of $\FF_2+u\FF_2$.}\label{table-1}
	\centering
	\begin{adjustbox}{max width=\textwidth}
	\footnotesize
	\setlength{\tabcolsep}{4pt}
	\begin{tabular}{cc}\midrule
	$\FF_2+u\FF_2$ & Symbol\\\midrule
	$0$   & \texttt{0}\\
	$1$   & \texttt{1}\\
	$u$   & \texttt{2}\\
	$1+u$ & \texttt{3}\\\midrule
	\end{tabular}
	\end{adjustbox}
	\end{table}
	
\subsection{New Self-Dual Code of Length 54}
	
The possible weight enumerators of a binary self-dual $[54,27,10]$ code are given in \cite{R-029} as
	\begin{align*}
	    W_{54,1}&=1+(351-8\alpha)x^{10}+(5031+24\alpha)x^{12}+\cdots,\\
	    W_{54,2}&=1+(351-8\alpha)x^{10}+(5543+24\alpha)x^{12}+\cdots,
	\end{align*}
where $\alpha\in\ZZ$. Previously known $\alpha$ values for weight enumerator $W_{54,1}$ can be found online at \cite{WEPD} (see \cite{R-029,R-056,R-203,R-038,R-132,R-201,R-044,R-202,R-114,R-195,R-194,R-183}).

We obtain one new optimal binary self-dual code of length 54 which has weight enumerator $W_{54,1}$ for
	\begin{enumerate}[label=]
        \item $\alpha=23$.
	\end{enumerate}	
	
The new code is constructed by first applying Theorem \ref{theorem-1} to obtain a code of length 54 over $\FF_2$ (Table \ref{table-54-1}) and then searching for neighbours of this code using Remark \ref{remark-1} (Table \ref{table-54-2}).

	\begin{table}
	\caption{Code of length 54 over $\FF_2$ from Theorem \ref{theorem-1} to which we apply Remark \ref{remark-1} to obtain the code in Table \ref{table-54-2}, where $\vctg{\xi}=(\xi_1,\xi_2,\xi_3,\xi_4,\xi_5,\xi_6)$.}\label{table-54-1}
	\centering
	\begin{adjustbox}{max width=\textwidth}
	\footnotesize
	\setlength{\tabcolsep}{4pt}
	\begin{tabular}{ccccc}\midrule
	$\mathcal{C}_{54,i}^*$ & $\vct{a}$ & $\vct{b}$ & $\vct{c}$ & $\vctg{\xi}$\\\midrule
	1 & \texttt{(0111000101101)} & \texttt{(1101110000100)} & \texttt{(0101111110011)} & \texttt{(001101)}\\\midrule
	\end{tabular}
	\end{adjustbox}
	\end{table}

	\begin{table}
	\caption{New binary self-dual $[54,27,10]$ code from searching for neighbours of $\mathcal{C}_{54,j}^*$ as given in Table \ref{table-54-1} using Remark \ref{remark-1} with $\vct{x}=(\vct{0},\vct{x}_0)$.}\label{table-54-2}
	\centering
	\begin{adjustbox}{max width=\textwidth}
	\footnotesize
	\setlength{\tabcolsep}{4pt}
	\begin{tabular}{cccccc}\midrule
	$\mathcal{C}_{54,i}$ & $\mathcal{C}_{54,j}^*$ & $\vct{x}_0$ & $W_{54,k}$ & $\alpha$ & $|\aut{\mathcal{C}_{54,i}}|$\\\midrule
	1 & 1 & \texttt{(000001100101001000111101101)} & 1 & $23$ & $3$\\\midrule
	\end{tabular}
	\end{adjustbox}
	\end{table}

\subsection{New Self-Dual Codes of Length 68}

The possible weight enumerators of a binary self-dual $[68,34,12]$ code are given in \cite{R-132} as
	\begin{align*}
	    W_{68,1}&=1+(442+4\alpha)x^{12}+(10864-8\alpha)x^{12}+\cdots,\\
	    W_{68,2}&=1+(442+4\alpha)x^{12}+(14960-8\alpha-256\beta)x^{12}+\cdots,
	\end{align*}
where $\alpha,\beta\in\ZZ$. Previously known $(\alpha,\beta)$ values for weight enumerators $W_{68,1}$ and $W_{68,2}$ can be found online at \cite{WEPD} (see \cite{R-029,R-057,R-198,R-132,R-199,R-136,R-196,R-197,R-048,R-126,R-127,R-122,R-072,R-200,R-176,R-071,R-042,R-008,R-046,R-043,R-108,R-106,R-025,R-080,R-097,R-105,R-107,R-096,R-103,R-113,R-110,R-099,R-101,R-104,R-100,R-067,R-111,R-151}).

We obtain 14 new extremal binary self-dual codes of length 68 of which 8 have weight enumerator $W_{68,1}$ for
	\begin{enumerate}[label=]
        \item $\alpha\in\{110,\lb113,\lb114,\lb116,\lb118,\lb121,\lb123,\lb124\}$
	\end{enumerate}	
and 6 have weight enumerator $W_{68,2}$ for
	\begin{enumerate}[label=]
	    \item $\beta=1$ and $\alpha\in\{20,\lb28,\lb32,\lb34,\lb36,\lb37\}$.
	\end{enumerate}	
	
Of the 14 new codes, 7 are constructed by applying Theorem \ref{theorem-1} over $\FF_2+u\FF_2$ (Table \ref{table-68-1}) and 7 are constructed by first applying Theorem \ref{theorem-1} to obtain a code of length 34 over $\FF_2+u\FF_2$ (Table \ref{table-68NBR-1}) and then searching for neighbours of the image of this code under $\varphi_{\FF_2+u\FF_2}$ using Remark \ref{remark-1} (Table \ref{table-68NBR-2}).	

	\begin{table}
	\caption{New binary self-dual $[68,34,12]$ codes from Theorem \ref{theorem-1} over $\FF_2+u\FF_2$, where $\vctg{\xi}=(\xi_1,\xi_2,\xi_3,\xi_4,\xi_5,\xi_6)$.}\label{table-68-1}
	\centering
	\begin{adjustbox}{max width=\textwidth}
	\footnotesize
	\setlength{\tabcolsep}{4pt}
	\begin{tabular}{ccccccccccc}\midrule
	$\mathcal{C}_{68,i}$ & $\lambda$ & $\mu$ & $\vct{a}$ & $\vct{b}$ & $\vct{c}$ & $\vctg{\xi}$ & $W_{68,j}$ & $\alpha$ & $\beta$ & $|\aut{\mathcal{C}_{68,i}}|$\\\midrule
	1 & \texttt{1} & \texttt{1} & \texttt{(22120031)} & \texttt{(02331100)} & \texttt{(33331213)} & \texttt{(101132)} & 1 & $110$ & $-$ & $2$\\
	2 & \texttt{1} & \texttt{1} & \texttt{(10021300)} & \texttt{(31232012)} & \texttt{(30313131)} & \texttt{(120023)} & 1 & $124$ & $-$ & $2$\\
	3 & \texttt{1} & \texttt{1} & \texttt{(01323103)} & \texttt{(20022123)} & \texttt{(00300222)} & \texttt{(013332)} & 2 & $20$ & $1$ & $2$\\
	4 & \texttt{1} & \texttt{3} & \texttt{(01230200)} & \texttt{(13010312)} & \texttt{(22003002)} & \texttt{(102232)} & 2 & $28$ & $1$ & $2$\\
	5 & \texttt{1} & \texttt{1} & \texttt{(31221023)} & \texttt{(30003111)} & \texttt{(13012103)} & \texttt{(233310)} & 2 & $32$ & $1$ & $2$\\
	6 & \texttt{1} & \texttt{1} & \texttt{(03210210)} & \texttt{(32221121)} & \texttt{(13331101)} & \texttt{(122201)} & 2 & $34$ & $1$ & $2$\\
	7 & \texttt{1} & \texttt{1} & \texttt{(00030320)} & \texttt{(21031233)} & \texttt{(32100012)} & \texttt{(122201)} & 2 & $36$ & $1$ & $2$\\\midrule
	\end{tabular}
	\end{adjustbox}
	\end{table}

	\begin{table}
	\caption{Code of length 34 over $\FF_2+u\FF_2$ from Theorem \ref{theorem-1} to the image of which under $\varphi_{\FF_2+u\FF_2}$ we then apply Remark \ref{remark-1} to obtain the codes in Table \ref{table-68NBR-2}, where $\vctg{\xi}=(\xi_1,\xi_2,\xi_3,\xi_4,\xi_5,\xi_6)$.}\label{table-68NBR-1}
	\centering
	\begin{adjustbox}{max width=\textwidth}
	\footnotesize
	\setlength{\tabcolsep}{4pt}
	\begin{tabular}{ccccccc}\midrule
	$\mathcal{C}_{34,i}^*$ & $\lambda$ & $\mu$ & $\vct{a}$ & $\vct{b}$ & $\vct{c}$ & $\vctg{\xi}$\\\midrule
	1 & \texttt{1} & \texttt{1} & \texttt{(01323103)} & \texttt{(20022123)} & \texttt{(00300222)} & \texttt{(013332)}\\\midrule
	\end{tabular}
	\end{adjustbox}
	\end{table}
	
	\begin{table}
	\caption{New binary self-dual $[68,34,12]$ codes from searching for neighbours of $\varphi_{\FF_2+u\FF_2}(\mathcal{C}_{34,j}^*)$ using Remark \ref{remark-1} with $\vct{x}=(\vct{0},\vct{x}_0)$, where $\mathcal{C}_{34,j}^*$ are as given in Table \ref{table-68NBR-1}.}\label{table-68NBR-2}
	\centering
	\begin{adjustbox}{max width=\textwidth}
	\footnotesize
	\setlength{\tabcolsep}{4pt}
	\begin{tabular}{ccccccc}\midrule
	$\mathcal{C}_{68,i}$ & $\mathcal{C}_{34,j}^*$ & $\vct{x}_0$ & $W_{68,k}$ & $\alpha$ & $\beta$ & $|\aut{\mathcal{C}_{68,i}}|$\\\midrule
	8 & 1 & \texttt{(0101010011111010001101100011011100)} & 1 & $113$ & $-$ & $1$\\
	9 & 1 & \texttt{(1110010011100001110010110111100100)} & 1 & $114$ & $-$ & $1$\\
	10 & 1 & \texttt{(1010100100010111000000100111010111)} & 1 & $116$ & $-$ & $1$\\
	11 & 1 & \texttt{(0011000011011101010101010100010000)} & 1 & $118$ & $-$ & $1$\\
	12 & 1 & \texttt{(0101010001111010000101100011011111)} & 1 & $121$ & $-$ & $1$\\
	13 & 1 & \texttt{(0011001001011000000110010111110101)} & 1 & $123$ & $-$ & $1$\\
	14 & 1 & \texttt{(0101110101111010001101100011011101)} & 2 & $37$ & $1$ & $1$\\\midrule
	\end{tabular}
	\end{adjustbox}
	\end{table}
	
\subsection{New Self-Dual Codes of Length 82}

The possible weight enumerators of a binary self-dual $[82,41,14]$ code are given in \cite{R-138} as
	\begin{align*}
	    W_{82,1}&=1+560x^{14}+60724x^{16}+233545x^{18}+\cdots,\\
	    W_{82,2}&=1+(3280+2\alpha)x^{14}+(36244-2\alpha+128\beta)x^{16}\\&\quad+(506153-26\alpha-896\beta)x^{18}+\cdots,\\
	    W_{82,3}&=1+(3280+2\alpha)x^{14}+(36244-2\alpha+128\beta)x^{16}\\&\quad+(514345-26\alpha-896\beta)x^{18}+\cdots,
	\end{align*}
where $\alpha,\beta\in\ZZ$. Previously known $(\alpha,\beta)$ values for weight enumerators $W_{82,2}$ and $W_{82,3}$ can be found online at \cite{WEPD} (see \cite{R-044,R-085,R-138}).

We obtain 18 new best known binary self-dual codes of length 82 of which 7 have weight enumerator $W_{82,2}$ for
	\begin{enumerate}[label=]
        \item $\beta=18$ and $\alpha\in\{-2z:z=331,\lb344,\lb353,\lb357,\lb367,\lb368,\lb369\}$
	\end{enumerate}	
and 11 have weight enumerator $W_{82,3}$ for
	\begin{enumerate}[label=]
	    \item $\beta=0$ and $\alpha\in\{-2z:z=388,\lb389,\lb393,\lb399,\lb406,\lb408,\lb414\}$;
	    \item $\beta=1$ and $\alpha\in\{-2z:z=409\}$;
	    \item $\beta=2$ and $\alpha\in\{-2z:z=409,\lb419\}$;
	    \item $\beta=5$ and $\alpha\in\{-2z:z=427\}$.
	\end{enumerate}

The new codes are constructed by applying Theorem \ref{theorem-1} over $\FF_2$ (Table \ref{table-82}).

	\begin{table}
	\caption{New binary self-dual $[82,41,14]$ codes from Theorem \ref{theorem-1} over $\FF_2$, where $\vctg{\xi}=(\xi_1,\xi_2,\xi_3,\xi_4,\xi_5,\xi_6)$.}\label{table-82}
	\centering
	\begin{adjustbox}{max width=\textwidth}
	\footnotesize
	\setlength{\tabcolsep}{4pt}
	\begin{tabular}{ccccc}\midrule
	$\mathcal{C}_{82,i}$ & $\vct{a}$ & $\vct{b}$ & $\vct{c}$ & $\vctg{\xi}$\\\midrule
	1 & \texttt{(00110011100000000110)} & \texttt{(00100110011101010011)} & \texttt{(00010010010001000001)} & \texttt{(101010)}\\
	2 & \texttt{(11001011011010110101)} & \texttt{(10011010011011010000)} & \texttt{(01010011100101001010)} & \texttt{(101010)}\\
	3 & \texttt{(00011110011001011110)} & \texttt{(01010101010011110100)} & \texttt{(10101110111000111011)} & \texttt{(101010)}\\
	4 & \texttt{(00000110100111111111)} & \texttt{(00110110000111101000)} & \texttt{(11111011010111011000)} & \texttt{(101001)}\\
	5 & \texttt{(11100011011110101011)} & \texttt{(11110001101100110011)} & \texttt{(00100010100000001010)} & \texttt{(101010)}\\
	6 & \texttt{(11111110010110010010)} & \texttt{(10001001101001001110)} & \texttt{(01111010111110011001)} & \texttt{(101001)}\\
	7 & \texttt{(00111010001011010100)} & \texttt{(11001010111101110001)} & \texttt{(10001100011010110001)} & \texttt{(101010)}\\
	8 & \texttt{(00110011011011110001)} & \texttt{(00101110100101000100)} & \texttt{(10110001110000000001)} & \texttt{(101110)}\\
	9 & \texttt{(10000011001000100011)} & \texttt{(00110001010001110100)} & \texttt{(00010001110001000101)} & \texttt{(101101)}\\
	10 & \texttt{(11101110100101100010)} & \texttt{(01110011001100110001)} & \texttt{(00010100000110011010)} & \texttt{(101101)}\\
	11 & \texttt{(00011011111101000011)} & \texttt{(11000000001100111001)} & \texttt{(10100000101010010010)} & \texttt{(101110)}\\
	12 & \texttt{(00011110101110000110)} & \texttt{(11000011010011000101)} & \texttt{(01001010001111101110)} & \texttt{(101110)}\\
	13 & \texttt{(00100000101100010000)} & \texttt{(11010101010010100011)} & \texttt{(01011101110000111001)} & \texttt{(101101)}\\
	14 & \texttt{(10001111010001011100)} & \texttt{(00000001010010011000)} & \texttt{(01101011111010000110)} & \texttt{(101101)}\\
	15 & \texttt{(10011111001010110001)} & \texttt{(11000010101110010110)} & \texttt{(01000011001011110111)} & \texttt{(101110)}\\
	16 & \texttt{(11100100001011100001)} & \texttt{(00101100110000110100)} & \texttt{(00011111001001111100)} & \texttt{(101101)}\\
	17 & \texttt{(10001110110000101100)} & \texttt{(00111010000111110010)} & \texttt{(01110111101001100001)} & \texttt{(101110)}\\
	18 & \texttt{(00001101111100100101)} & \texttt{(00011001110100011111)} & \texttt{(01001100001011101111)} & \texttt{(101110)}\\\midrule
	\end{tabular}
	\end{adjustbox}
	\end{table}
	
	\addtocounter{table}{-1}
	\begin{table}
	\caption{(continued)}
	\centering
	\begin{adjustbox}{max width=\textwidth}
	\footnotesize
	\setlength{\tabcolsep}{4pt}
	\begin{tabular}{ccccc}\midrule
	$\mathcal{C}_{82,i}$ & $W_{82,j}$ & $\alpha$ & $\beta$ & $|\aut{\mathcal{C}_{82,i}}|$ \\\midrule
	1 & 2 & $-738$ & $18$ & $1$\\
	2 & 2 & $-736$ & $18$ & $1$\\
	3 & 2 & $-734$ & $18$ & $1$\\
	4 & 2 & $-714$ & $18$ & $1$\\
	5 & 2 & $-706$ & $18$ & $1$\\
	6 & 2 & $-688$ & $18$ & $1$\\
	7 & 2 & $-662$ & $18$ & $1$\\
	8 & 3 & $-828$ & $0$ & $1$\\
	9 & 3 & $-816$ & $0$ & $1$\\
	10 & 3 & $-812$ & $0$ & $1$\\
	11 & 3 & $-798$ & $0$ & $1$\\
	12 & 3 & $-786$ & $0$ & $1$\\
	13 & 3 & $-778$ & $0$ & $1$\\
	14 & 3 & $-776$ & $0$ & $1$\\
	15 & 3 & $-818$ & $1$ & $1$\\
	16 & 3 & $-838$ & $2$ & $1$\\
	17 & 3 & $-818$ & $2$ & $1$\\
	18 & 3 & $-854$ & $5$ & $1$\\\midrule
	\end{tabular}
	\end{adjustbox}
	\end{table}

\subsection{New Self-Dual Codes of Length 94}

The possible weight enumerators of a binary self-dual $[94,47,16]$ code are given in \cite{R-184} as
	\begin{align*}
	    W_{94,1}&=1+2\alpha x^{16}+(134044-2\alpha+128\beta)x^{18}\\&\quad+(2010660-30\alpha-896\beta)x^{20}+\cdots,\\
	    W_{94,2}&=1+2\alpha x^{16}+(134044-2\alpha+128\beta)x^{18}\\&\quad+(2018852-30\alpha-896\beta)x^{20}+\cdots,\\
	    W_{94,3}&=1+2\alpha x^{16}+(134044-2\alpha+128\beta)x^{18}\\&\quad+(2190884-30\alpha-896\beta)x^{20}+\cdots,
	\end{align*}
where $\alpha,\beta\in\ZZ$. Previously known $(\alpha,\beta)$ values for weight enumerator $W_{94,1}$ can be found online at \cite{WEPD} (see \cite{R-184,AMR4}).

We obtain 12 new best known binary self-dual codes of length 94 which have weight enumerator $W_{94,1}$ for
	\begin{enumerate}[label=]
        \item $\beta=-92$ and $\alpha\in\{46z:z=101\}$;
        \item $\beta=-46$ and $\alpha\in\{46z:z=75,\lb80,\lb82,\lb91\}$;
        \item $\beta=-23$ and $\alpha\in\{46z:z=64,\lb80\}$;
        \item $\beta=0$ and $\alpha\in\{46z:z=51,\lb55,\lb56,\lb76,\lb78\}$.
	\end{enumerate}	

The new codes are constructed by applying Theorem \ref{theorem-1} over $\FF_2$ (Table \ref{table-94}).

	\begin{table}
	\caption{New binary self-dual $[94,47,16]$ codes from Theorem \ref{theorem-1} over $\FF_2$, where $\vctg{\xi}=(\xi_1,\xi_2,\xi_3,\xi_4,\xi_5,\xi_6)$.}\label{table-94}
	\centering
	\begin{adjustbox}{max width=\textwidth}
	\footnotesize
	\setlength{\tabcolsep}{4pt}
	\begin{tabular}{ccccc}\midrule
	$\mathcal{C}_{94,i}$ & $\vct{a}$ & $\vct{b}$ & $\vct{c}$ & $\vctg{\xi}$\\\midrule
	1 & \texttt{(01111111111001110101110)} & \texttt{(01101101000111011010001)} & \texttt{(00001000000000000000000)} & \texttt{(001110)}\\
	2 & \texttt{(10010111111101010000010)} & \texttt{(11100100111001001111001)} & \texttt{(00001000000000000000000)} & \texttt{(001110)}\\
	3 & \texttt{(01100111001001011111010)} & \texttt{(10110101001111101000010)} & \texttt{(11010111010100010110011)} & \texttt{(001110)}\\
	4 & \texttt{(10010101100111000001101)} & \texttt{(11010000110110110000001)} & \texttt{(01010001111011001010111)} & \texttt{(110010)}\\
	5 & \texttt{(00000111101001000010100)} & \texttt{(11110100110110100111000)} & \texttt{(01001111001111101100100)} & \texttt{(001101)}\\
	6 & \texttt{(11011110010100111000000)} & \texttt{(01110100011001101101111)} & \texttt{(01110000001111000111111)} & \texttt{(001101)}\\
	7 & \texttt{(01011011110110010001110)} & \texttt{(10010110110110001100101)} & \texttt{(00000100000000000000000)} & \texttt{(110010)}\\
	8 & \texttt{(01100001100001100101010)} & \texttt{(11111101000110000010101)} & \texttt{(00100000000000000000000)} & \texttt{(001101)}\\
	9 & \texttt{(00000111001111011011110)} & \texttt{(11100000000100010011010)} & \texttt{(01101111110111000010001)} & \texttt{(110010)}\\
	10 & \texttt{(01101101011111000010001)} & \texttt{(10100110011101001101101)} & \texttt{(01011000110000010010101)} & \texttt{(110010)}\\
	11 & \texttt{(11010010011100001111011)} & \texttt{(10001110000000010001110)} & \texttt{(11101110011100011101000)} & \texttt{(110010)}\\
	12 & \texttt{(10101100011011001010111)} & \texttt{(00010010000011111000010)} & \texttt{(00111100000011101111110)} & \texttt{(001101)}\\\midrule	
	\end{tabular}
	\end{adjustbox}
	\end{table}
	
	\addtocounter{table}{-1}
	\begin{table}
	\caption{(continued)}
	\centering
	\begin{adjustbox}{max width=\textwidth}
	\footnotesize
	\setlength{\tabcolsep}{4pt}
	\begin{tabular}{ccccc}\midrule
	$\mathcal{C}_{94,i}$ & $W_{94,j}$ & $\alpha$ & $\beta$ & $|\aut{\mathcal{C}_{94,i}}|$ \\\midrule
	1 & 1 & $4646$ & $-92$ & $2\cdot 23$\\
	2 & 1 & $3450$ & $-46$ & $2\cdot 23$\\
	3 & 1 & $3680$ & $-46$ & $23$\\
	4 & 1 & $3772$ & $-46$ & $23$\\
	5 & 1 & $4186$ & $-46$ & $23$\\
	6 & 1 & $2944$ & $-23$ & $23$\\
	7 & 1 & $3680$ & $-23$ & $23$\\
	8 & 1 & $2346$ & $0$ & $2\cdot 23$\\
	9 & 1 & $2530$ & $0$ & $23$\\
	10 & 1 & $2576$ & $0$ & $23$\\
	11 & 1 & $3496$ & $0$ & $23$\\
	12 & 1 & $3588$ & $0$ & $23$\\\midrule
	\end{tabular}
	\end{adjustbox}
	\end{table}
	
\section{Conclusion}

In this work, we defined a modification of a previously given bordered matrix construction for self-dual codes which utilises $\lambda$-circulant matrices. We proved the necessary conditions required by the construction to produce self-dual codes over finite commutative Frobenius rings of characteristic 2. We demonstrated the ability of this technique by using it along with the well-known neighbour construction to produce the following new singly-even binary self-dual codes:
	
\begin{enumerate}
\item[]\textbf{Code of length 54:} We were able to construct a new singly-even binary self-dual $[54,27,10]$ code which has weight enumerator $W_{54,1}$ for:
	\begin{align*}
		\alpha=23.
	\end{align*}

\item[]\textbf{Codes of length 68:} We were able to construct new binary self-dual $[68,34,12]$ codes which have weight enumerator $W_{68,1}$ for:
	\begin{align*}
		\alpha\in\{110,\lb113,\lb114,\lb116,\lb118,\lb121,\lb123,\lb124\}
	\end{align*}
and weight enumerator $W_{68,2}$ for:
	\begin{align*}
		\beta=1 \ \text{and} \ \alpha\in\{20,\lb28,\lb32,\lb34,\lb36,\lb37\}.
	\end{align*}

\item[]\textbf{Codes of length 82:} We were able to construct new binary self-dual $[82,41,14]$ codes which have weight enumerator $W_{82,2}$ for:
	\begin{align*}
		\beta=18 \ \text{and} \ \alpha\in\{-2z:z=331,\lb344,\lb353,\lb357,\lb367,\lb368,\lb369\}
	\end{align*}
and weight enumerator $W_{82,3}$ for:
	\begin{align*}
		&\beta=0 \ \text{and} \ \alpha\in\{-2z:z=388,\lb389,\lb393,\lb399,\lb406,\lb408,\lb414\},\\
		&\beta=1 \ \text{and} \ \alpha\in\{-2z:z=409\},\\
		&\beta=2 \ \text{and} \ \alpha\in\{-2z:z=409,\lb419\},\\
		&\beta=5 \ \text{and} \ \alpha\in\{-2z:z=427\}.
	\end{align*}

\item[]\textbf{Codes of length 94:} We were able to construct new binary self-dual $[94,47,16]$ codes which have weight enumerator $W_{94,1}$ for:
	\begin{align*}
		&\beta=-92 \ \text{and} \ \alpha\in\{46z:z=101\},\\
		&\beta=-46 \ \text{and} \ \alpha\in\{46z:z=75,\lb80,\lb82,\lb91\},\\
		&\beta=-23 \ \text{and} \ \alpha\in\{46z:z=64,\lb80\},\\
		&\beta=0 \ \text{and} \ \alpha\in\{46z:z=51,\lb55,\lb56,\lb76,\lb78\}.
	\end{align*}
\end{enumerate}

Due to the size of the search field for the given construction, all of the codes were obtained by random searches. As such, if a more comprehensive generation procedure was implemented, there would likely be more new codes which could be constructed with this technique. A suggestion for future work could to be investigate further modification or generalisation of our construction. We could also consider our construction after substituting the matrix $X$ with some other orthogonal matrix, for example, an orthogonal matrix arising from group rings. Similarly, we could replace the matrix $C$ with another orthogonal matrix. This could possibly lead to the discovery of new binary self-dual codes with more atypically structured automorphism groups.

\bibliographystyle{plainnat}
\bibliography{paper5}
\end{document}